\documentclass[letterpaper, 10 pt, conference]{ieeeconf}  

\IEEEoverridecommandlockouts  
\usepackage{textcomp}

\usepackage{amsmath,amssymb}

\usepackage{amsthm}
\usepackage{amsfonts,mathrsfs}
\usepackage{color}
\usepackage{graphicx}
\usepackage{epsfig}
\usepackage{subcaption}
\usepackage[utf8]{inputenc}
\usepackage{amsmath}
\let\labelindent\relax
\usepackage{enumitem}
\usepackage{accents}
\usepackage{mathtools} 
\usepackage{tabularx}
\newcolumntype{Y}{>{\centering\arraybackslash}X} 
\usepackage[prependcaption,colorinlistoftodos]{todonotes}
\usepackage{cancel}
\usepackage[normalem]{ulem} 
\usepackage{diagbox}
\usepackage{pgfplots}
\pgfplotsset{compat=1.11}
\usepackage{subcaption}
\usepgfplotslibrary{fillbetween} 

\usepackage[flushmargin]{footmisc}

\usepackage{algorithm,algorithmicx}
\usepackage{algpseudocode}
\usepackage[hyphens]{url} 
\usepackage[hidelinks]{hyperref}

\definecolor{CBcyan}{rgb}{0.4,0.8,0.9333}    
\definecolor{CBred}{rgb}{0.9333,0.4,0.4667}  

\usepackage{tikz}
\usetikzlibrary{arrows.meta} 
\usetikzlibrary{fit, positioning}
\usetikzlibrary{calc}
\usetikzlibrary{decorations.pathmorphing} 

\tikzstyle{block} = [rectangle, minimum width=.7cm, minimum height=.7cm, text centered, draw=black]
\tikzstyle{tallblock} = [rectangle, minimum width=.5cm, minimum height=1cm, text centered, draw=black]
\tikzstyle{line} = [thick,-,>=stealth]
\tikzstyle{arrow} = [thick,->,>=stealth]
\tikzstyle{roundedblock} = [rectangle, minimum width=4cm, minimum height=2cm, text centered, draw=black, rounded corners=0.2cm]

\algrenewcommand\textproc{\textsc}

\DeclareMathOperator*{\argmin}{arg\,min}

\newcommand{\mc}{\mathcal}
\newcommand{\bb}{\mathbb}
\newcommand{\R}{\bb R}

\DeclareMathAlphabet{\mathbbmsl}{U}{bbm}{m}{sl}

\newcommand{\Rmnum}[1]{\expandafter\@slowromancap\romannumeral #1@}
\newcommand{\eod}{\ensuremath{\hfill\Box}}
\newcommand{\qedd}{\ensuremath{\hfill \blacksquare}}

\renewcommand{\tilde}[1]{\accentset{\sim}{#1}}

\makeatletter
\newcommand{\specialcell}[1]{\ifmeasuring@#1\else\omit$\displaystyle#1$\ignorespaces\fi}
\makeatother

\newtheorem{lemma}{Lemma}

\newtheorem{theorem}{Theorem}
\newtheorem{definition}{Definition}
\newtheorem{assumption}{Assumption}

\newtheorem{remark}{Remark}

\newcommand{\AVI}{\mathrm{AVI}}

\definecolor{set19c1}{HTML}{E41A1C}
\definecolor{set19c2}{HTML}{377EB8}
\definecolor{set19c3}{HTML}{4DAF4A}
\definecolor{set19c4}{HTML}{984EA3}
\definecolor{set19c5}{HTML}{FF7F00}
\definecolor{set19c6}{HTML}{FFFF33}
\definecolor{set19c7}{HTML}{A65628}
\definecolor{set19c8}{HTML}{F781BF}
\definecolor{set19c9}{HTML}{999999}

\usepackage{balance}

\title{DR-DAQP: A Hybrid Operator-Splitting and Active-Set \\Solver for Affine Variational Inequalities}
\author{Daniel Arnstr\"om, Emilio Benenati, Giuseppe Belgioioso
\thanks{D. Arnstr\"om contributed to this work in an independent capacity, E. Benenati and G. Belgioioso are with the Department of Decision and Control Systems, KTH Royal Institute of Technology, Stockholm, Sweden. Emails: \texttt{daniel.arnstrom@gmail.com}, \texttt{\{benenati, giubel\}@kth.se}.}
\thanks{This work was partially supported by the Wallenberg AI, Autonomous Systems and Software Program (WASP) funded by the Knut and Alice Wallenberg Foundation.}
}
\date{}

\begin{document}

\maketitle

\begin{abstract}
We present \texttt{DR-DAQP}, an open-source solver for strongly monotone affine variational inequalities (AVIs) that combines Douglas--Rachford operator splitting with an active-set acceleration strategy.
The key idea is to estimate the active set along the iterations to attempt directly computing a point that satisfies the optimality conditions. This step yields the exact AVI solution when the active set is correctly estimated, thus overcoming the asymptotic convergence limitation inherent in first-order methods.
Moreover, we exploit warm-starting and pre-factorization of relevant matrices to further accelerate evaluation of the algorithm iterations.
We prove convergence and establish conditions under which the algorithm terminates in finite time with the exact solution.
Numerical experiments on a randomly generated set of AVIs and on a set of multi-agent optimal control problems show that \texttt{DR-DAQP} outperforms several state-of-the-art solvers.
A high-performing C implementation is available at \url{https://github.com/darnstrom/daqp}, with easily-accessible interfaces to Julia, MATLAB, and Python.
\end{abstract}

\section{Introduction}
Variational inequalities (VIs) provide a general mathematical framework that encompasses convex optimization and complementarity problems~\cite{facchinei2003finite}. Under suitable assumptions, generalized Nash equilibrium problems (GNEPs) can also be formulated as VIs~\cite{facchinei_generalized_2010}, making this framework particularly relevant to the control of multi-agent systems with coupled constraints and potentially conflicting objectives. Relevant applications include autonomous driving~\cite{fridovich-keil_efficient_2020,le_cleach_lucidgames_2021}, autonomous racing~\cite{fieni_game_2025}, and swarm robotics~\cite{do_nascimento_game_2023}.
\par In this paper, we focus on the subclass of strongly monotone \emph{affine variational inequality} (AVIs). These arise naturally in GNEPs where the agents' cost functions are linear-quadratic and the constraints polyhedral, and serve as subproblems in sequential approximation methods for more general nonlinear games \cite{fridovich-keil_efficient_2020,zhu_sequential_2023}.
Despite the extensive literature on AVI solution algorithms \cite{facchinei2003finite,dirkse1995path,cao1996pivotal,eckstein1998operator,kim2018structure,baghbadorani2025douglas,benenati2025explicit}, the availability of readily-usable implementations focused on reliability and speed of computation is crucially limited, as we review next.

\subsection*{Related work on AVI algorithms and solvers}

Existing approaches to solving AVIs can broadly be classified into three categories. First, \emph{pivotal methods}~\cite{cao1996pivotal} generalize the simplex method to AVIs. 
Prominent implementations include the \texttt{PATH} solver~\cite{dirkse1995path} and its AVI-tailored variant \texttt{PATHAVI}~\cite{kim2018structure}. These methods can compute exact solutions, but their computational cost can become substantial as the problem size increases~\cite{murty2009computational}.

Second, \emph{operator-splitting methods}~\cite{eckstein1998operator,baghbadorani2025douglas} decompose the VI into simpler subproblems and construct a fixed-point iteration whose limit yields a solution. As convergence is merely asymptotic, obtaining high-accuracy solutions may require many iterations. This constitutes a crucial limitation for real-time control, where both solution accuracy and constraint satisfaction directly affect closed-loop performance~\cite{liu2024input}. The \texttt{monviso} package~\cite{mignoni_monviso_2025} provides implementations of several such algorithms. 

Third, in the parametric setting, the solution to the AVI can be computed in closed-form~\cite{benenati2025explicit} using the parametric solver \texttt{PDAQP}~\cite{arnstrom2024high}. While this explicit approach allows one to find a solution for all values of a parameter, it is limited to problems of modest size, as the polyhedral partition that defines the piece-wise affine solution map exhibits a combinatorial growth.  

\subsection*{Contributions}
In this paper, we present \texttt{DR-DAQP}, a hybrid AVI solver that bridges the gap between operator-splitting and active-set methods. Our contributions are threefold:
\begin{enumerate}
    \item We develop a hybrid algorithm that enhances the Douglas--Rachford splitting with acceleration strategies based on active-set identification (Algorithm~\ref{alg:dr-daqp}). We prove global convergence under strong monotonicity and, under additional linear independence and strict complementarity assumptions, finite termination to the \emph{exact} AVI solution  (Theorems~\ref{thm:convergence} and~\ref{thm:finite}).
    \item We provide an efficient \texttt{C} implementation that applies the dual active-set QP solver \texttt{DAQP}~\cite{arnstrom2022dual} on the projection subproblems. The implementation exploits active-set warm starts and reuses matrix factorizations across iterations to reduce computational overhead.
    \item We evaluate \texttt{DR-DAQP} on random AVI benchmarks, where it achieves up to two orders of magnitude speedup over \texttt{PATH}, and on a game-theoretic MPC case study, where it outperforms a mixed-integer approach in \cite{bemporad2026nashopt} by several orders of magnitude.
\end{enumerate}

A detailed user guide for the \texttt{C}, \texttt{MATLAB}, \texttt{Python}, and \texttt{Julia} interfaces can be found in the \texttt{DAQP} documentation\footnote{\url{https://darnstrom.github.io/daqp/start/advanced/avi}}.


\section{Problem Formulation}\label{sec:problem}

\subsection{Affine variational inequalities}

We consider affine variational inequalities (AVIs), in which the objective is to find a vector $x\in \mathcal{X}$ such that
\begin{equation}
    \label{eq:avi}
    \quad (Hx+f)^\top(y-x) \geq 0, \quad \forall y\in \mathcal{X}, 
\end{equation}
where $H\in\R^{n\times n}$, $f\in\R^n$, and $\mc X$ is polyhedral, namely
\begin{equation}
    \mc X:=\{x\in\R^n \mid Ax\leq b\},
\end{equation}
with $A\in\R^{m\times n}$ and $b\in\R^m$. Throughout the paper, we impose the following assumption.

\begin{assumption}\label{ass:pd}
    $H + H^\top$ is positive definite.
\end{assumption}
Assumption~\ref{ass:pd} is equivalent to strong monotonicity
of the affine operator $x \mapsto Hx+f$~\cite[Prop.~2.3.3]{facchinei2003finite}.
Since $\mathcal{X}$ is nonempty, closed, and convex,
it ensures that the AVI in~\eqref{eq:avi} admits a unique
solution~\cite[Thm.~2.3.3]{facchinei2003finite}.

A point $x$ solves~\eqref{eq:avi} if and only if there exists
a dual variable $\lambda \in \mathbb{R}^m$ satisfying the following
Karush--Kuhn--Tucker (KKT) conditions~\cite[Thm.~1.2.1]{facchinei2003finite}:
\begin{subequations}
    \label{eq:kkt}
  \begin{align}
    H x + f + A^T \lambda &= 0, \label{eq:kkt:stat} \\
    0 \leq \lambda \perp b-Ax &\geq 0. \label{eq:kkt:compl}
  \end{align}
\end{subequations}
Despite the structural similarity to standard QP optimality conditions, the KKT system in~\eqref{eq:kkt} differs in a crucial way: $H$ is in general \emph{not} symmetric and, thus, standard QP solvers cannot be directly employed to find a solution.

\subsection{Active-set structure}\label{sec:activeset}

Our approach exploits the active-set structure of the AVI. Let $x^*$ denote the solution to~\eqref{eq:avi},
and define its active set as
\begin{equation}\label{eq:active_set}
    \mc A:=\{ i \in \{1,\ldots,m\} \mid A_i x^* = b_i \},
\end{equation}
where $A_i$ and $b_i$ denote the $i$-th row of $A$ and $i$-th element of $b$, respectively. We denote as $\overline{\mc A}$ the complement of $\mc A$, that is, $\overline{\mc A}:=\{1,\dots,m\} \setminus \mc A$.
The complementarity condition~\eqref{eq:kkt:compl} implies that a dual solution $\lambda^*$ satisfies $\lambda^*_i = 0$ for all $i \in \overline{\mc A}$, so the stationarity condition~\eqref{eq:kkt:stat} reduces to 
\begin{equation}\label{eq:reduced_stat_cond}
    Hx^* + f + A_{\mc A}^\top \lambda^*_{\mc A} = 0,
\end{equation}
where $A_{\mc A}$ and $\lambda^*_{\mc A}$ denote the submatrix and subvector with row/element indices in $\mc A$. Combining~\eqref{eq:reduced_stat_cond} with the active-set definition~\eqref{eq:active_set}, we observe that the primal-dual pair $(x^*, \lambda^*_{\mc A})$ can be computed from the linear system
\begin{equation}
    \label{eq:kkt-system}
  \begin{bmatrix}
      H & A_{\mathcal{A}}^\top \\
      A_{\mathcal{A}} & 0
  \end{bmatrix}
  \begin{bmatrix}
   x \\ \lambda_{\mc A}  
  \end{bmatrix} = 
  \begin{bmatrix}
      -f \\ b_{\mathcal{A}}
  \end{bmatrix}.
\end{equation}
We conclude that, \emph{if the active set $\mc A$ is known}, the AVI reduces to solving the linear system~\eqref{eq:kkt-system} and verifying primal feasibility ($Ax^* \leq b$) and dual feasibility ($\lambda^*_{\mc A}\geq 0$). In practice, the active set $\mc A$ is not known a priori. Our strategy, detailed in the next section, is to iteratively estimate $\mc A$ and $x^*$ along the iterations of a Douglas--Rachford (DR) splitting method, and use the corresponding reduced system to attempt exact solution recovery or a Newton-type correction.

\section{The \texttt{DR-DAQP} Algorithm}\label{sec:algorithm}

\subsection{Douglas--Rachford splitting for AVIs}\label{sec:DR_for_AVIs}
{The AVI in \eqref{eq:avi} is equivalent to the inclusion \cite[Eq. 1.1.3]{facchinei2003finite}
\begin{equation}\label{eq:zero_find_operator}
	0\in\mathrm{N}_{\mc X}(x) + Hx + f,
\end{equation}
where $\mathrm{N}_\mc{X}$ is the normal cone of $\mc X$ \cite[Eq. 1.1.2]{facchinei2003finite}. Let us define the matrices $H_s$ and $\tilde{H}$ as 
\begin{equation}\label{eq:splitting}
	H_s := \tfrac{1}{2}(H+H^\top), \quad \tilde{H} := \rho I + H_s,
\end{equation}
with $\rho > 0$ being a user-defined parameter. We solve \eqref{eq:zero_find_operator} with the DR iteration \cite[\S 2.2]{eckstein1998operator} on the splitting  
\begin{align}
	\begin{split}
		T_1(x) &:= \mathrm{N}_{\mc X}(x) + \tfrac{H_s}{2}x + f \\
		T_2(x) &:= (H-\tfrac{H_s}{2})x
	\end{split}
\end{align}
with preconditioning matrix $Q$ such that $Q^{-\top}Q^{-1} = \rho I + \frac{H_s}{2}$. The notation for $T_1, T_2$ and $Q$ follows that in \cite{eckstein1998operator}, to which we refer for details on the derivation of the DR iteration.} 
By introducing the parameter-dependent linear term
\begin{equation}\label{eq:ftilde}
    \tilde{f}(z) \triangleq f + \left(H-\tilde{H}\right) z,
\end{equation}
the DR iteration can be expressed as in Algorithm~\ref{alg:dr}. Note that, at each iteration, Step~\ref{step:qp-solve} of Algorithm~\ref{alg:dr} requires the solution of the convex QP parametric in $z$:
\begin{equation}
    \label{eq:dr-qp}
  \begin{aligned}
      \underset{x}{\text{minimize}}\:\: &\tfrac{1}{2}  x^\top \tilde{H} x + \tilde{f}(z)^\top x, \\
      \text{subject to}\:\: & A x \leq b.
  \end{aligned}
\end{equation}
{The QP in \eqref{eq:dr-qp} can be interpreted as a symmetrized Newton step \cite[\S 3.3.1]{patriksson_nonlinear_1999} with additional regularization given by the term $\rho I$.}
Across the iterations of Algorithm~\ref{alg:dr}, the only term of the QP that varies is the linear part of the objective function, {while the constraint matrix $A$ and Hessian $\tilde{H}$ remain identical. Thus, we can exploit the ability of the dual active-set QP solver \texttt{DAQP}~\cite{arnstrom2022dual} to cache and reuse its internal $LDL^\top$ factorizations, and to warm-start each QP solve with the active set of the previous iteration.}
\begin{algorithm}[t]
    \caption{Douglas--Rachford splitting (\texttt{DR}) for~\eqref{eq:avi}}
    \label{alg:dr}
    \begin{algorithmic}[1]
        \Require $\AVI(H,f,A,b)$, tolerance $\eta > 0$, parameter $\rho > 0$.   
        \Ensure Approximate solution $y$ 
        \State $H_s \leftarrow \textbf{sym}(H)$;\quad $\tilde{H} \leftarrow \rho I + H_s$ 
        \Repeat 
        \State $y_k \leftarrow  \argmin_{Ax\leq b}\, \tfrac{1}{2}x^\top\tilde{H}x + \tilde{f}(z_k)^\top x$ \label{step:qp-solve}
        \If {$\|y_k-z_k\| \leq \eta$} \textbf{return} $y_k$ 
        \EndIf
        \State $z_{k+1} \leftarrow (\rho I + H)^{-1}\left(\rho y_k+ H z_k + \tfrac{H_s}{2}(y_k-z_k)\right)$
        \State $k\leftarrow k+1$
        \Until{convergence}
  \end{algorithmic}
\end{algorithm}

\subsection{Active-set acceleration: the \texttt{DR-DAQP} algorithm}

Algorithm~\ref{alg:dr} converges to the AVI solution only asymptotically. To accelerate convergence and enable finite termination,
we enhance it with corrections based on the active-set structure described in Section~\ref{sec:activeset}. The resulting method is summarized in Algorithm~\ref{alg:dr-daqp} and its key steps are listed next:
\begin{enumerate}[leftmargin=*]
    \item \textbf{Active-set identification.} At each iteration, the dual active-set QP solver \texttt{DAQP} provides both the solution $y_k$ and the active set $\mathcal{A}_k$ of the subproblem~\eqref{eq:dr-qp}. We use $\mathcal{A}_k$ as an estimate of the active set of the AVI solution.
    \item \textbf{Exact solution evaluation.} When the active set stabilizes, namely, $\mc A_k = \mc A_{k-1}$, we solve the reduced KKT system~\eqref{eq:kkt-system} for a candidate solution $(\tilde{z}_k,\lambda_k)$. If $\tilde{z}_k$ is primal-feasible and $\lambda_k \geq 0$, the algorithm returns the \emph{exact} solution {up to computer precision.}
    \item \textbf{Newton step acceleration.} If the solution to the linear system~\eqref{eq:kkt-system} is either not primal or dual feasible, we attempt a Newton-type step to accelerate convergence (Step \ref{step:QP}). This is detailed in Section~\ref{sec:merit}.
    \item \textbf{Warm-started QP solution.} As explained in Section \ref{sec:DR_for_AVIs}, we solve the QP subproblems using \texttt{DAQP} with cached matrix factorizations and warm-start from the previous solution, yielding substantial computational savings.
\end{enumerate}

\begin{algorithm}[t]
    \caption{\texttt{DR-DAQP}: Hybrid active-set method for~\eqref{eq:avi}}
    \label{alg:dr-daqp}
    \begin{algorithmic}[1]
        \Require $\AVI(H,f,A,b)$, parameter $\rho > 0$
        \Ensure Primal-dual solution $(z,\lambda)$
        \State $H_s \leftarrow \text{\bf{sym}}(H)$;\quad $\tilde{H} \leftarrow \rho I + H_s$;\quad $\delta \leftarrow \infty$
        \Repeat
        \State $y_k,\mathcal{A}_k \leftarrow$  solve~\eqref{eq:dr-qp} for $z = z_{k}$ using \texttt{DAQP} \label{step2:qp-solve}
        \If{$\mathcal{A}_k = \mathcal{A}_{k-1}$} \label{step:as-test} \Comment{Active set stabilized}
        \State $\tilde{z}_k,\lambda_k \leftarrow$ solve KKT system~\eqref{eq:kkt-system} for $\mathcal{A} = \mathcal{A}_k$ \label{step:kkt_system} 
        \If{$\lambda_k \geq 0$ \textbf{and} $ A \tilde{z}_k \leq b$} \label{step:as-terminate} 
        \State \textbf{return} $\tilde{z}_k, \lambda_k$ \Comment{Exact solution found}
        \EndIf
        \State $\tilde{y}_k,{\mathcal{A}_k} \leftarrow$  solve~\eqref{eq:dr-qp} for $z = \tilde{z}_{k}$ using \texttt{DAQP} \label{step:QP}
        \If{$\|\tilde{y}_k-\tilde{z}_k\| < \delta $} \Comment{Newton step accepted} \label{step:dr_condition}
        \State $y_k \leftarrow \tilde{y}_k$;\quad $z_k \leftarrow \tilde{z}_k$;\quad $\delta \leftarrow \|\tilde{y}_k - \tilde{z}_k \|$ \label{step:heur-update}
        \EndIf
        \EndIf
        \State $z_{k+1} \leftarrow (\rho I + H)^{-1}\left(\rho y_k+ H z_k + \tfrac{H_s}{2}(y_k-z_k)\right)$  \label{step2:DR_update}
        \State $k\leftarrow k+1$
        \Until{convergence}
  \end{algorithmic}
\end{algorithm}

\subsection{Merit function and Newton-step acceptance}\label{sec:merit}

When the candidate solution obtained at Step \ref{step:kkt_system} of Algorithm \ref{alg:dr-daqp} by solving~\eqref{eq:kkt-system} does not satisfy primal or dual feasibility (as is the case when the estimate $\mc A_k$ is different from $\mc A$), the algorithm attempts a Newton-type step. We use the \emph{natural residual} as a merit function to evaluate the consequent improvement.

\begin{definition}[Natural residual]\label{def:natres}
    For a symmetric positive definite matrix $Q$, the $Q$-weighted projection of $v$ onto $\mathcal{X}$ is
    \begin{equation}
        \label{eq:proj}
        \Pi_{\mathcal{X},Q}(v) \triangleq \argmin_{x\in \mathcal{X}}\|x - v\|_Q.
    \end{equation}
    The $Q$-weighted natural residual of the AVI~\eqref{eq:avi} of $z$ is
    \begin{equation}\label{eq:natres}
    R_{\mathcal{X},Q}(z) \triangleq z - \Pi_{\mathcal{X},Q}\left(z- Q^{-1}(Hz+f)\right).
    \end{equation}
\end{definition}


The natural residual provides a measure of the quality of the solution estimate \cite[\S 6]{facchinei2003finite} and it is zero at the AVI solution \cite[Prop.~1.5.8]{facchinei2003finite}. We leverage this fact when evaluating the Newton step, which is accepted only if it strictly reduces the residual compared to all previously accepted Newton steps (line~\ref{step:heur-update}), thus ensuring monotone improvement of the solution estimates. {As we detail in the proof of Theorem~\ref{thm:convergence}, this choice ensures that the Newton step is only performed a finite number of times, and thus it does not alter the asymptotic convergence properties of the DR iterations.}  Moreover, we leverage a connection with the QP in \eqref{eq:dr-qp} that enables one to compute the residual with minimal overhead: Let $y^*$ be the solution of \eqref{eq:dr-qp} for a given parameter $z$. Then, it can be immediately verified that $\|R_{\mathcal{X},\tilde{H}}(z)\| = \|z-{y}^*\| $.
{\begin{remark}\label{rem:rho}
	The parameter $\rho$ is inversely related to $\|Q\|$, which, in turn, influences the strong monotonicity and Lipschitz constants of the preconditioned operators that define the DR iteration, see \cite[Eq.~29]{eckstein1998operator}. The resulting effect of $\rho$ on the convergence rate can be characterized using the bounds in \cite[Thm. 6.5]{giselsson_tight_2017}. In our numerical experiments, we choose $\rho=\|H\|_F$, which provides a good balance between convergence speed and numerical conditioning.
\end{remark}}

\subsection{Convergence Analysis}\label{sec:convergence}

We now establish the convergence properties of Algorithm~\ref{alg:dr-daqp}. First, we show that it retains the global asymptotic convergence of the nominal DR method. All proofs are provided {in the appendix.}
\begin{theorem}[Convergence]\label{thm:convergence}
    Under Assumption~\ref{ass:pd}, the sequence $\{z_k\}$ generated by Algorithm~\ref{alg:dr-daqp} converges to the solution of the AVI in~\eqref{eq:avi}. \eod
\end{theorem}
{Under additional assumptions, we prove} that it terminates in a finite number of iterations with the \emph{exact} solution.
\begin{assumption}[LICQ]\label{ass:licq}
    The rows of $A_{\mathcal{A}}$ are linearly independent, where $\mc A$ is defined in \eqref{eq:active_set}. \eod
\end{assumption}
LICQ implies uniqueness of the dual solution to \eqref{eq:kkt} \cite[Prop. 3.2.1]{facchinei2003finite}, which we denote as $\lambda^*$.
\begin{assumption}[Strict complementarity]\label{ass:scs}
    \[
    \lambda^*_i > 0, ~ \forall i \in \mathcal{A}. \tag*{$\square$}
    \]
\end{assumption}
\begin{theorem}[Finite-time termination]\label{thm:finite}
    Let Assumptions~\ref{ass:pd},~\ref{ass:licq} and~\ref{ass:scs} hold true. Then, Algorithm~\ref{alg:dr-daqp} terminates with the exact solution $(x^*, \lambda^*)$ in a finite number of iterations.
\end{theorem}

\begin{remark}[Stabilization of the active set]\label{rem:stab}
{In Algorithm~\ref{alg:dr-daqp}, the active set is considered stabilized once $\mathcal{A}_k = \mathcal{A}_{k-1}$. The same convergence analysis applies when the active set is considered stable after it has remained unchanged for $K$ consecutive iterations. In practice, choosing a higher $K$ yields a more conservative strategy in attempting Newton corrections, each of which requires an additional QP solve. In our experiments, we choose $K=5$.}
\end{remark}

\section{Numerical Experiments}\label{sec:experiments}

We evaluate \texttt{DR-DAQP} on two classes of problems: randomly generated AVIs of varying size and asymmetry, and a game-theoretic MPC benchmark. All experiments\footnote{Code for reproducing all of the conducted experiments is available at \url{https://github.com/darnstrom/dr-daqp-cdc26}.} are performed on an AMD Ryzen 7 8840HS processor.

\subsection{Randomly generated AVIs}

We generate AVIs of the form~\eqref{eq:avi} with  
\begin{equation}\label{eq:Hgen}
    H =  (1-\gamma^{\text{asym}}) \frac{M_s^\top M_s}{\|M_s^\top M_s\|} + \gamma^{\text{asym}} \frac{M_a^\top-M_a}{\|M_a^\top-M_a\|},
\end{equation}
where the elements of $A$, $M_s$, and $M_a$ are drawn  from $\mathcal{N}(0,1)$, $b$ is chosen to ensure feasibility, and the parameter $\gamma^{\text{asym}}\in[0,1]$ controls the degree of asymmetry of $H$.
\subsubsection{Effect of the QP solver}
We first investigate the impact of the QP solver used in line~\ref{step2:qp-solve} of Algorithm~\ref{alg:dr-daqp}. We vary the number of decision variables $n\in\{5,10,\ldots,50\}$ and generate 100~AVIs per problem size with $\gamma^{\text{asym}}=0.5$ and $m=10n$ constraints. Figure~\ref{fig:clara} compares the total solve time when using \texttt{DAQP} (warm-started and cold-started) and \texttt{Clarabel}~\cite{goulart2024clarabel} as the QP subroutine. The warm-started \texttt{DAQP} variant achieves up to two orders of magnitude speedup over \texttt{Clarabel}, validating the importance of exploiting the fixed QP structure.

\begin{figure}[H]
    \centering
    \begin{tikzpicture}[scale=1]
        \pgfplotstableread{data/daqp_vs_daqpcold_clara.dat}{\claradata}
        \usepgfplotslibrary{fillbetween} 
        \begin{axis}[
            xmin=5,xmax=50,
            ymode=log,
            xlabel={Decision variable dimension $n$},
            ylabel={Total Solve time [s]},
            legend style={at ={(0.5,1.05)},anchor=south}, 
            legend style={nodes={scale=0.8, transform shape}},
            legend cell align={left},legend columns=3,
            ymajorgrids,yminorgrids,xmajorgrids,
            y post scale=0.5,
            ]
            \addplot [set19c1,very thick] table [x={n}, y={tmeddaqp}] {\claradata}; 
            \addplot [set19c3,very thick] table [x={n}, y={tmeddaqpcold}] {\claradata}; 
            \addplot [set19c2,very thick] table [x={n}, y={tmedclara}] {\claradata}; 

            \addplot [set19c1,dotted, name path=daqpmax] table [x={n}, y={tmaxdaqp}] {\claradata}; 
            \addplot [set19c1,dotted, name path=daqpmin] table [x={n}, y={tmindaqp}] {\claradata}; 
            \addplot [set19c1,fill opacity=0.15] fill between [of=daqpmax and daqpmin];

            \addplot [set19c2,dotted, name path=claramax] table [x={n}, y={tmaxclara}] {\claradata}; 
            \addplot [set19c2,dotted, name path=claramin] table [x={n}, y={tminclara}] {\claradata}; 
            \addplot [set19c2,fill opacity=0.15] fill between [of=claramax and claramin];

            \addplot [set19c3,dotted, name path=daqpcoldmax] table [x={n}, y={tmaxdaqpcold}] {\claradata}; 
            \addplot [set19c3,dotted, name path=daqpcoldmin] table [x={n}, y={tmindaqpcold}] {\claradata}; 
            \addplot [set19c3,fill opacity=0.15] fill between [of=daqpcoldmax and daqpcoldmin];

            \legend{\:\texttt{DAQP} (warm),\texttt{DAQP} (cold),\texttt{Clarabel}}
        \end{axis}
    \end{tikzpicture}
    \caption{Total solve time when solving the subproblems \eqref{eq:dr-qp} in line~\ref{step:qp-solve} using \texttt{DAQP} (warm and cold started) and \texttt{Clarabel}.}
    \label{fig:clara}
\end{figure}

\subsubsection{Comparison with \texttt{PATH}}
Next, we compare Algorithm~\ref{alg:dr-daqp} against the state-of-the-art pivotal solver \texttt{PATH}~\cite{dirkse1995path}\footnote{{Specifically, \texttt{PATHSolver.jl}} v1.7.9}, and against the standard QP solver Clarabel applied to the reformulation of the AVI as a QP in a lifted dual space \cite{aghassi2006solving}. {First, we compare how the solution time scales with the number of problem variables by generating 100 AVIs for each problem size, considering $n\in\{25,50,75,100\}$,} $\gamma^{\text{asym}} = 0.5$, and $m=10n$. Figure~\ref{fig:path} shows that \texttt{DR-DAQP} is approximately one order of magnitude faster in median solve time and up to two orders of magnitude faster in worst-case solve time compared to \texttt{PATH} {for all tested problem sizes.} The lifted QP approach \cite{aghassi2006solving}, despite allowing the use of a standard QP solver, incurs significant overhead due to the increased dimension of the lifted problem, as well as worse problem conditioning (as the resulting QP is merely convex even under Assumption \ref{ass:pd}). {Second, we compare the effect of the degree of asymmetry of $H$ on the solution time, considering $n=100$, $\gamma^{\text{asym}} \in [0.1,0.9]$, and $m=100$. Figure~\ref{fig:sweep_gamma} shows that the solution time of \texttt{DR-DAQP} increases with $\gamma^{\text{asym}}$, while remaining below that of \texttt{PATH} and a lifted QP formulation in all tested cases. }

\begin{figure}[t]
  \centering
      \begin{tikzpicture}[scale=1]
        \pgfplotstableread{data/path_vs_daqp_vs_qp.dat}{\pathdata}
        \begin{axis}[
            title = {},
            xmin=25,xmax=100,
            ymode=log,
            xlabel={Decision variable dimension $n$},
            ylabel={Solve time [s]},
            legend style={at ={(0.0,1.0)},anchor=north west}, 
            legend style={nodes={scale=0.7, transform shape}},
            ymajorgrids,yminorgrids,xmajorgrids,
            y post scale=0.6,
            ]
            \addplot [set19c1,very thick] table [x={n}, y={tmeddaqp}] {\pathdata}; 
            \addplot [set19c2,very thick] table [x={n}, y={tmedpath}] {\pathdata}; 
            \addplot [set19c3,very thick] table [x={n}, y={tmedqp}] {\pathdata}; 

            \addplot [set19c1,dotted, name path=A] table [x={n}, y={tmaxdaqp}] {\pathdata}; 
            \addplot [set19c1,dotted, name path=B] table [x={n}, y={tmindaqp}] {\pathdata}; 
            \addplot [set19c1,fill opacity=0.15] fill between [of=A and B];

            \addplot [set19c2,dotted, name path=Ac] table [x={n}, y={tmaxpath}] {\pathdata}; 
            \addplot [set19c2,dotted, name path=Bc] table [x={n}, y={tminpath}] {\pathdata}; 
            \addplot [set19c2,fill opacity=0.15] fill between [of=Ac and Bc];

            \addplot [set19c3,dotted, name path=Aq] table [x={n}, y={tmaxqp}] {\pathdata}; 
            \addplot [set19c3,dotted, name path=Bq] table [x={n}, y={tminqp}] {\pathdata}; 
            \addplot [set19c3,fill opacity=0.15] fill between [of=Aq and Bq];

            \legend{\texttt{DR-DAQP},\texttt{PATH}, Lifted QP}
        \end{axis}
    \end{tikzpicture}
  \caption{Solve time comparison of \texttt{DR-DAQP}, \texttt{PATH}, and a lifted QP formulation on randomly generated AVIs with $\gamma^{\text{asym}}=0.5$ and $m=10n$. Solid lines show the median; shaded regions span the min-max range over 100 instances.}
  \label{fig:path}
\end{figure}
\begin{figure}[t]
	\centering
        \begin{tikzpicture}[scale=1]
        \pgfplotstableread{data/path_vs_daqp_asymmetry.dat}{\aspathdata}
        \begin{axis}[
            title = {},
            xmin=0.05,xmax=0.95,
            ymode=log,
            xlabel={Asymmetry factor $\gamma$},
            ylabel={Solve time [s]},
            legend style={at ={(0.5,1.05)},anchor=south}, 
            legend style={nodes={scale=0.8, transform shape}},
            legend cell align={left},legend columns=3,
            ymajorgrids,yminorgrids,xmajorgrids,
            y post scale=0.4,
            ]
            \addplot [set19c1,very thick] table [x={asymmetry_factor}, y={tmeddaqp}] {\aspathdata}; 
            \addplot [set19c2,very thick] table [x={asymmetry_factor}, y={tmedpath}] {\aspathdata}; 
            \addplot [set19c3,very thick] table [x={asymmetry_factor}, y={tmedqp}] {\aspathdata}; 

            \addplot [set19c1,dotted, name path=A] table [x={asymmetry_factor}, y={tmaxdaqp}] {\aspathdata}; 
            \addplot [set19c1,dotted, name path=B] table [x={asymmetry_factor}, y={tmindaqp}] {\aspathdata}; 
            \addplot [set19c1,fill opacity=0.15] fill between [of=A and B];

            \addplot [set19c2,dotted, name path=Ac] table [x={asymmetry_factor}, y={tmaxpath}] {\aspathdata}; 
            \addplot [set19c2,dotted, name path=Bc] table [x={asymmetry_factor}, y={tminpath}] {\aspathdata}; 
            \addplot [set19c2,fill opacity=0.15] fill between [of=Ac and Bc];

            \addplot [set19c3,dotted, name path=Aq] table [x={asymmetry_factor}, y={tmaxqp}] {\aspathdata}; 
            \addplot [set19c3,dotted, name path=Bq] table [x={asymmetry_factor}, y={tminqp}] {\aspathdata}; 
            \addplot [set19c3,fill opacity=0.15] fill between [of=Aq and Bq];

            \legend{\texttt{DR-DAQP},\texttt{PATH}, Lifted QP}
        \end{axis}
    \end{tikzpicture}
	\caption{Solve time comparison of \texttt{DR-DAQP}, \texttt{PATH}, and a lifted QP formulation on randomly generated AVIs with $n=100$ and $m=1000$. Solid lines show the median; shaded regions span the min-max range over 100 instances.}
	\label{fig:sweep_gamma}
\end{figure}

\subsubsection{Convergence behavior}
Finally, we examine the convergence behavior by comparing the residual $\|x_k - x^*\|$ across iterations for different methods. Figure~\ref{fig:iter} compares: (i)~standard DR iterations, (ii)~\texttt{DR-DAQP} without Newton steps, (iii)~\texttt{DR-DAQP} with Newton steps, and (iv)~ {Forward-Backward iterations} (from \texttt{monviso}~\cite{mignoni_monviso_2025}). {In accordance to Theorem~\ref{thm:finite}, \texttt{DR-DAQP} reaches the exact solution in a finite number of iterations, whereas the other methods converge only asymptotically. The Newton-step acceleration enables \texttt{DR-DAQP} to reach the solution in fewer iterations, at the cost of an additional QP solve per attempted Newton step.}

\begin{figure}[t]
  \centering
      \begin{tikzpicture}[scale=1]
        \pgfplotstableread{data/iter_comparison.dat}{\iterdata}
        \begin{axis}[
            title = {},
            xmin=1,xmax=100,
            ymin = 0.0001,
            ymode=log,
            xlabel={$k$ (\# of iterations)},
            ylabel={$\|x_k - x^*\|$},
            legend style={at ={(1.0,1.0)},anchor=north east}, 
            legend style={nodes={scale=0.7, transform shape}},
            legend cell align={left},
            legend columns=2,
            ymajorgrids,yminorgrids,xmajorgrids,
            y post scale=0.6,
            ]
            \addplot [set19c4,very thick] table [x={iter}, y={dr}] {\iterdata}; 
            \addplot [set19c3,very thick] table [x={iter}, y={daqpnonewt}] {\iterdata}; 
            \addplot [set19c1,very thick] table [x={iter}, y={daqp}] {\iterdata}; 
            \addplot [set19c2,very thick] table [x={iter}, y={pg}] {\iterdata};

            \legend{\texttt{DR}, \texttt{DR-DAQP} (no Newton), \texttt{DR-DAQP}, Projected gradient}
        \end{axis}
    \end{tikzpicture}
  \caption{Reduction of the residual $\|x_k - x^*\|$ over iterations for a representative AVI instance. \texttt{DR-DAQP} with Newton steps achieves finite termination, while first-order methods converge only asymptotically.}
  \label{fig:iter}
\end{figure}

\subsection{Game-theoretic MPC}\label{sec:gtmpc-exp}

A key motivation for studying AVIs arises in game-theoretic MPC, which is an emerging control architecture for non-cooperative multi-agent systems \cite{fridovich-keil_efficient_2020, le_cleach_lucidgames_2021, do_nascimento_game_2023}. The game-theoretic MPC formulation requires the solution of a generalized Nash equilibrium (GNE) at each time step to compute the control inputs: When the cost functions are quadratic and the dynamics and constraints are linear, the GNE problem can be written as an AVI of the form~\eqref{eq:avi}. As the finite-horizon control problem needs to be solved in real time, obtaining a fast solution is crucial to achieve a satisfactory controller sampling rate. A detailed derivation of the AVI formulation can be found in \cite{benenati_linear-quadratic_2026}.

{To demonstrate the practical relevance of \texttt{DR-DAQP}, we benchmark against supported open-source solvers in \texttt{NashOpt}~\cite{bemporad2026nashopt}, namely, QP-GNEP, operator extrapolation, Lemke's method, and Log-IPM. We compare the solvers on randomly generated game-theoretic MPC problems. Each problem comprises three agents, each controlling one input of a randomly generated six-state, three-output linear system $x_{k+1}=Ax_k+Bu_k$, $y_k=[I_3\;0]x_k$, with $A$ scaled to exhibit a spectral radius of $0.85$.  The agents use distinct quadratic output-tracking weights $Q_{y,i}=0.05I_3+0.20e_i e_i^\top$, input-increment weight $5$, input bounds $-2\leq u_k\leq2$, increment bounds $-0.5\leq\Delta u_k\leq0.5$, and softened output bounds $-3\leq y_k\leq3$.  We consider prediction horizons $T\in\{10,15,20,25,30\}$, corresponding to a number of decision variables between $33$ and $93$  and a number of constraints between $183$ and $543$, after substituting the dynamics to eliminate the state variables.  We report median online wall-clock time, including numerical GNE construction and solution but excluding offline MPC condensation, with interquartile bands over the $50$ instances. The experiments show that \texttt{DR-DAQP} outperforms all other solvers in all tested instances.}

\providecommand{\mpcdata}{data/new_mpc}
\begin{figure}
    \centering
    \begin{tikzpicture}[scale=1]
        \begin{axis}[
            ymode=log,
            xmin=10,xmax=30,
            xlabel={Prediction horizon},
            ylabel={Solution time (s)},
            grid=both,
            y post scale=0.6,
            legend style={at={(0.5,1.05)}, anchor=south, legend columns=3},
            legend style={nodes={scale=0.8, transform shape}},
            ]
            \addplot[name path=dr_daqpLower, draw=none, forget plot]
            table[x=horizon, y=q25] {\mpcdata/dr_daqp.dat};
            \addplot[name path=dr_daqpUpper, draw=none, forget plot]
            table[x=horizon, y=q75] {\mpcdata/dr_daqp.dat};
            \addplot[fill=set19c1, fill opacity=0.15, draw=none, forget plot]
            fill between[of=dr_daqpUpper and dr_daqpLower];
            \addplot[color=set19c1, very thick, mark=*]
            table[x=horizon, y=median] {\mpcdata/dr_daqp.dat};
            \addlegendentry{\texttt{DR-DAQP}}      

            \addplot[name path=lemkeLower, draw=none, forget plot]
            table[x=horizon, y=q25] {\mpcdata/lemke.dat};
            \addplot[name path=lemkeUpper, draw=none, forget plot]
            table[x=horizon, y=q75] {\mpcdata/lemke.dat};
            \addplot[fill=set19c3, fill opacity=0.15, draw=none, forget plot]
            fill between[of=lemkeUpper and lemkeLower];
            \addplot[color=set19c3, very thick, mark=*]
            table[x=horizon, y=median] {\mpcdata/lemke.dat};
            \addlegendentry{\texttt{Lemke}}
            
            \addplot[name path=log_ipmLower, draw=none, forget plot]
            table[x=horizon, y=q25] {\mpcdata/log_ipm.dat};
            \addplot[name path=log_ipmUpper, draw=none, forget plot]
            table[x=horizon, y=q75] {\mpcdata/log_ipm.dat};
            \addplot[fill=set19c2, fill opacity=0.15, draw=none, forget plot]
            fill between[of=log_ipmUpper and log_ipmLower];
            \addplot[color=set19c2, very thick, mark=*]
            table[x=horizon, y=median] {\mpcdata/log_ipm.dat};
            \addlegendentry{\texttt{Log-IPM}}
            
            \addplot[name path=qp_gnepLower, draw=none, forget plot]
            table[x=horizon, y=q25] {\mpcdata/qp_gnep.dat};
            \addplot[name path=qp_gnepUpper, draw=none, forget plot]
            table[x=horizon, y=q75] {\mpcdata/qp_gnep.dat};
            \addplot[fill=set19c5, fill opacity=0.15, draw=none, forget plot]
            fill between[of=qp_gnepUpper and qp_gnepLower];
            \addplot[color=set19c5, very thick, mark=*]
            table[x=horizon, y=median] {\mpcdata/qp_gnep.dat};
            \addlegendentry{\texttt{QP-GNEP}}
            
            \addplot[name path=op_extrapLower, draw=none, forget plot]
            table[x=horizon, y=q25] {\mpcdata/op_extrap.dat};
            \addplot[name path=op_extrapUpper, draw=none, forget plot]
            table[x=horizon, y=q75] {\mpcdata/op_extrap.dat};
            \addplot[fill=set19c4, fill opacity=0.15, draw=none, forget plot]
            fill between[of=op_extrapUpper and op_extrapLower];
            \addplot[color=set19c4, very thick, mark=*]
            table[x=horizon, y=median] {\mpcdata/op_extrap.dat};
            \addlegendentry{\texttt{Op-Extrap}}

        \end{axis}
    \end{tikzpicture}
    \caption{Solve time for randomly generated game-theoretic MPC problems. Shaded areas show the inter-quartile bands. }
    \label{fig:nashopt}
\end{figure}

 

\subsection{Application to autonomous driving}
{We test \texttt{DR-DAQP} on the receding-horizon computation of a control input for an autonomous overtake maneuver.} The two vehicles are modeled as discretized unicycles \cite[Ex. 9.17]{sastry_nonlinear_2013} with sampling rate $10$Hz. Their objectives are designed to track a constant reference speed $(\overline{v}_i)_{i=1,2}$ and a time-varying reference lateral position $(\overline{l}_i)_{i=1,2}$:
\begin{align*}
J_1 &= \textstyle\sum_{t=1}^{T} \|v_1(t) - \overline{v}_1\|^2 + \|l_1 - \overline{l}_1(t)\|^2 + \|u_1(t)\|^2_{R} \\
J_2 &= \textstyle\sum_{t=1}^{T} \|v_2(t) - \overline{v}_2\|^2 + \|l_2 - \overline{l}_2(t)\|^2 + \|u_2(t)\|^2_{R},
\end{align*}
where $v_i$ and $l_i$ denote the speed and lateral position of vehicle $i$, and $u_i$ denotes its input (that is, acceleration and steering angle).
We set $\overline{v}_1 = 15 \mathrm{\frac{m}{s}},~\overline{v}_2 = 21 \mathrm{\frac{m}{s}}$, and $\overline{l}_1$ to maintain the right lane for all $t$. Instead, $\overline{l}_2(t)$ is set to track the left lane if  $9\geq p_1(t) - p_2(t) \geq 0 $, where $p_i$ denotes the position  of agent $i$, and the right lane otherwise. We impose saturation constraints on $v_i, l_i, u_i$, as well as non-convex collision-avoidance constraints
\begin{equation} \label{eq:collision_constraint}
    \frac{(p_1 - p_2)^2}{\Delta p^2_{\text{min}}} + \frac{(l_1 - l_2)^2}{\Delta l^2_{\text{min}}} \leq 1.
\end{equation}
\par The model is nonlinear both in dynamics and constraints. {At each time-step: \emph{(i)} We propagate the state trajectory by simulating the unicycle dynamics from the current state using the input sequence computed at the previous time step (at the initial time step, a zero input sequence is used);
\emph{(ii)} We linearize the dynamics and constraints around the resulting state trajectory, yielding a linear-quadratic dynamic game;
\emph{(iii)} We formulate the computation of the open-loop Nash equilibrium input of the resulting approximated game over a control horizon of $T=10$ as an AVI~\cite{benenati_linear-quadratic_2026};
\emph{(iv)} We solve the AVI using \texttt{DR-DAQP};
\emph{(v)} We apply the first input of the resulting equilibrium input sequence to the system and repeat the procedure at the next time step.}\\
We observe that the system successfully completes the maneuver with an AVI solution time of $2.1\cdot 10^{-4}\pm 4.1\cdot10^{-5} s$.
\begin{figure}
    \centering
    \includegraphics[width=\linewidth, trim=0 2.5cm 7cm 2cm, clip]{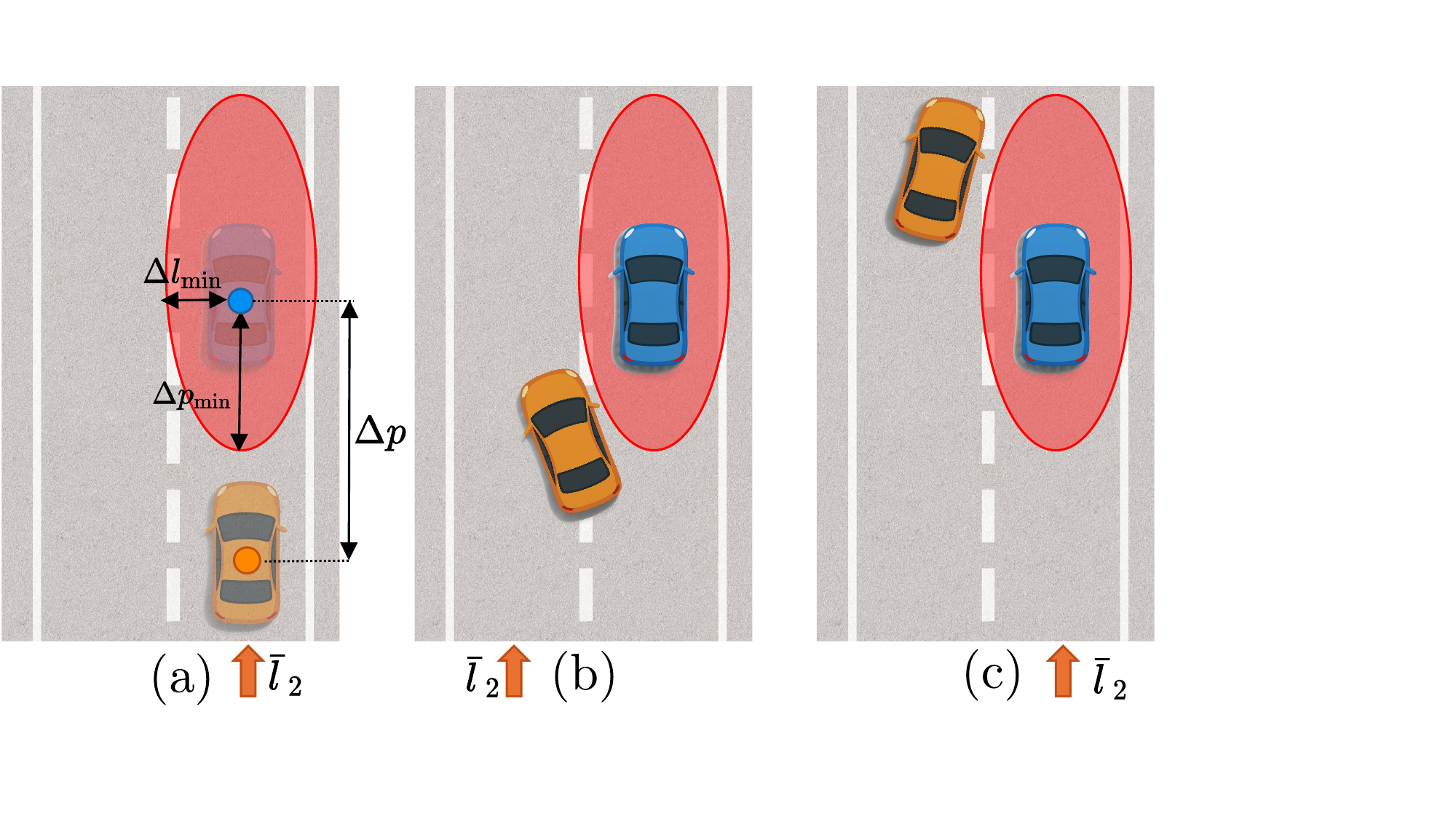}
    \caption{Simulated driving scenario. The safety distance constraint is drawn in red. When $\Delta p$ is below a threshold, $\bar{l}_2$ is set to the left lane \emph{(b)}. When $\Delta p<0$, $\bar{l}_2$ is set again to the right lane to complete the overtake \emph{(c)}. \\
    \emph{Video:} {\footnotesize\url{https://youtu.be/xoxPtcbNL3w}}}
    \label{fig:placeholder}
\end{figure}
\section{Conclusion}\label{sec:conclusion}

We presented \texttt{DR-DAQP}, a hybrid solver for affine variational inequalities that combines the global convergence of DR splitting with the finite-termination capability of active-set methods. By leveraging the fixed QP structure of the DR subproblems and the warm-starting capabilities of the \texttt{DAQP} solver, the proposed method achieves significant computational speedups, up to two orders of magnitude compared to the state-of-the-art solver \texttt{PATH} on random AVIs, and several orders of magnitude compared to mixed-integer approaches on game-theoretic MPC problems.

Future work includes extending the method to handle more general variational inequality structures (e.g., non-polyhedral constraint sets), investigating adaptive strategies for the regularization parameter~$\rho$, and integrating \texttt{DR-DAQP} into real-time control frameworks for multi-agent systems.

\section{Appendix}
We first introduce an auxiliary lemma:
\begin{lemma}\label{le:unique_kkt_primal}
    The system in \eqref{eq:kkt-system} admits at most one primal solution for each $\mc A$.
\end{lemma}
\begin{proof}
    Let us consider two solutions $(x',\lambda')$, $(x'',\lambda'')$ of \eqref{eq:kkt-system} for a given $\mc A$. By substituting into \eqref{eq:kkt-system} and taking the difference, one immediately obtains
    \begin{align}
        H(x'-x'') &= -A^\top_{\mc A} (\lambda'-\lambda'') \label{eq:unique_kkt_primal:1}\\
        A_{\mc A}(x'-x'') &= 0. \label{eq:unique_kkt_primal:2}
    \end{align}
    By inverting $H$ and substituting \eqref{eq:unique_kkt_primal:1} in \eqref{eq:unique_kkt_primal:2},
    \begin{equation}\label{eq:unique_kkt_primal:3}
        A_{\mc A}H^{-1}A^\top_{\mc A} (\lambda'-\lambda'') =0.
    \end{equation}
    We left-multiply the latter by $(\lambda'-\lambda'')^\top$ and use the fact that $H^{-1}$ is positive definite (as follows immediately from Assumption \ref{ass:pd}) to obtain
    \begin{align*}
         &(\lambda'-\lambda'')^\top A_{\mc A}H^{-1}A^\top_{\mc A} (\lambda'-\lambda'') =0 \\
         \Rightarrow~~ & A^\top_{\mc A} (\lambda'-\lambda'')=0
    \end{align*}
        Substituting the latter in \eqref{eq:unique_kkt_primal:1}, we conclude
        \begin{equation}
            (x'-x'') = -H^{-1}A^\top_{\mc A} (\lambda'-\lambda'') =0.
        \end{equation}
\end{proof}
\paragraph{Proof of Theorem \ref{thm:convergence}}  We argue that the implemented modifications to the DR algorithm (Steps~\ref{step:as-test}-\ref{step:heur-update}) do not compromise its convergence. \\
    Let $(y_k, z_k)_{k\in\mathbb{N}}$ be a sequence generated by Algorithm \ref{alg:dr-daqp}. We define the subset of iteration steps $\mc K_N \subseteq \mathbb{N}$ such that, at iteration $k\in\mc K_N$, Algorithm \ref{alg:dr-daqp} applies the Newton-type update, that is, Step \ref{step:heur-update}. Then, the condition in Step \ref{step:dr_condition} is verified for any $k\in\mc K_N$, which by definition of $\delta$ in Step \ref{step:heur-update} immediately leads to
    \begin{equation}\label{eq:newton_accepted_condition}
        \|\tilde{y}_k - \tilde{z}_k\| < \|\tilde{y}_p - \tilde{z}_p\| \qquad \forall k,p\in \mc K_N~\text{s.t.}~p< k.
    \end{equation}
    In words, the merit function $\|\tilde{y}_k - \tilde{z}_k\|$ for any iteration $k\in\mc K_N$ is strictly less than the merit function for any previous iteration $p\in \mc K_N$. We note that, for each $k\in\mc K_N$, the associated $\tilde{z}_k$ resulting from Step \ref{step:kkt_system} is the primal solution to the system of equations in \eqref{eq:kkt-system}, which is unique (Lemma \ref{le:unique_kkt_primal}) and it is in turn uniquely determined by the active set $\mc A_k$. Moreover, $\tilde{y}_k$ resulting from Step \ref{step:QP} is the unique solution to a QP identified by $\tilde{z}_k$. As a consequence, there is a unique pair $\tilde{z}_k, \tilde{y}_k$ for each combination of constraints $\mc A_k$. As the number of possible combinations of constraints is finite, the condition in \eqref{eq:newton_accepted_condition} can only be satisfied a finite number of times and, consequently, $\mc K_N$ is a finite set. Therefore, for each $k>\bar{k}$, where $\bar{k}$ is the maximum element of $\mc K_N$, Algorithm \ref{alg:dr-daqp} applies the Douglas--Rachford iteration (Steps \ref{step2:qp-solve} and \ref{step2:DR_update}), {which converges linearly to the solution of the AVI \cite[Thm. 6.5]{giselsson_tight_2017}.} 
\qedd 
    \paragraph{Proof of Theorem \ref{thm:finite}} Let us first introduce some preparatory results and notation. \\
    We denote $x^*, \lambda^*$ the primal-dual solution to the AVI in \eqref{eq:avi} (unique under Assumptions \ref{ass:pd}, \ref{ass:licq}). Let us define the functions $\mc S_{\text{x}}(z)$ and $\mc S_{\lambda}(z)$ that map the parameter $z$ respectively to the primal and dual solution of the QP in \eqref{eq:dr-qp}. Moreover, 
    we denote as $\mc A_{\text{QP}}(z)$ the set of active constraints at the solution of \eqref{eq:dr-qp}, that is,
    \begin{equation}
        A_i\mc S_{\text{x}}(z) = b_i ~\iff~ i\in\mc A_{\text{QP}}(z).
    \end{equation}
    Under Assumption \ref{ass:pd}, the primal solution is unique, thus $\mc S_\text{x}$ is well-defined. Uniqueness of the QP's dual solution is ensured by the following Lemma:
    \begin{lemma}[Transfer of regularity]\label{lem:transfer} Under Assumptions~\ref{ass:licq}, ~\ref{ass:scs}:
    \begin{enumerate}
        \item $\mc S_\text{x}(x^*) = x^*;~\mc S_\lambda(x^*) = \lambda^*$
        \item The QP in \eqref{eq:dr-qp} satisfies LICQ and strict complementarity.
    \end{enumerate}
\end{lemma}
\begin{proof}
    The KKT conditions for the QP in \eqref{eq:dr-qp} with parameter $z=x^*$ are 
    \begin{subequations} \label{eq:kkt_qp}
        \begin{align}
            &\tilde{H}x + f + (H-\tilde{H})x^* + A^T\lambda = 0, \\
            &0 \le \lambda \perp (b-Ax) \ge 0
        \end{align}
    \end{subequations}
    By substituting $x \leftarrow x^*$, $\lambda \leftarrow \lambda^*$ and comparing with the KKT conditions of the AVI in \eqref{eq:avi}, it is immediate to verify that $(x^*, \lambda^*)$ satisfies \eqref{eq:kkt_qp}, and thus it is a primal-dual solution to the QP problem in \eqref{eq:dr-qp}. 
    Since the constraints and the solution are identical for both the AVI and the QP, LICQ and strict complementarity of the QP directly follow from the assumptions.
\end{proof}  
We now show that Step \ref{step2:qp-solve} identifies the active set $\mc A$ defined in \eqref{eq:active_set} when $z_k$ lies in a neighborhood of the AVI solution. 
\begin{lemma} \label{le:critical_region}
    For all $z$ in a sufficiently small neighborhood of $x^*$, $\mc A_{\text{QP}}(z)\equiv \mc A$.
\end{lemma}
\begin{proof}
Following Lemma \ref{lem:transfer}, $\mc S_{\text{x}}(x^*)=x^*$. From \cite[Theorem 2]{bemporad_explicit_2002} and subsequent discussion, $\mc S_{\text{x}}(z), \mc S_{\lambda}(z)$ are affine maps and $\mc A_{\text{QP}}(z) = \mc A$ for all $z$ such that
\begin{subequations} \label{eq:active_constraints_condition}
    \begin{align}
        A_{\overline{\mc A}}\mc S_{\text{x}}(z) &\leq b_{\overline{\mc A}},\\
\mc [S_{\lambda}(z)]_{\mc A}& \geq 0,
    \end{align}
\end{subequations}
 where we recall that $\overline{\mc A}$ is the complement of $\mc A$, and we denote as $[S_{\lambda}(z)]_{\mc A}$ the elements of $S_{\lambda}(z)$ with index in $\mc A$. As \eqref{eq:active_constraints_condition} is satisfied strictly for $z=x^*$ from Assumption \ref{ass:scs}, and following the linearity of $\mc S_{\text{x}}, \mc S_{\lambda}$, we conclude that \eqref{eq:active_constraints_condition} holds true in a neighborhood of $x^*$.  \end{proof}
We now proceed with the proof of the main statement:
\begin{proof}
    Since the sequence $\{z_k\}$ generated by Algorithm 2 converges to $x^*$ (Theorem~\ref{thm:convergence}), there exists a finite iteration $k$ such that $z_k$ lies in a neighborhood of $x^*$ where $\mc A_{\text{QP}}(z_k) = \mc A$. This neighborhood exists following Lemma~\ref{le:critical_region}. Consequently, Step~\ref{step2:qp-solve} of the algorithm returns $\mc A$ and Step \ref{step:kkt_system} solves \eqref{eq:kkt-system}, yielding $(x^*, \lambda^*)$ that solves the AVI in \eqref{eq:avi} and triggering termination in Step~\ref{step:as-terminate}.
\end{proof}

    
    

\balance
\bibliographystyle{ieeetr}
\bibliography{lib.bib}

@article{cao1996pivotal,
    title={A pivotal method for affine variational inequalities},
    author={Cao, Menglin and Ferris, Michael C},
    journal={Mathematics of Operations Research},
    volume={21},
    number={1},
    pages={44--64},
    year={1996},
    publisher={INFORMS}
}

@article{giselsson_tight_2017,
	title = {Tight Global Linear Convergence Rate Bounds for {{Douglas}}--{{Rachford}} Splitting},
	author = {Giselsson, Pontus},
	year = 2017,
	journal = {Journal of Fixed Point Theory and Applications},
	volume = {19},
	number = {4},
	pages = {2241--2270},
	issn = {1661-7738, 1661-7746},
	doi = {10.1007/s11784-017-0417-1}
}

@inproceedings{fridovich-keil_efficient_2020,
  title = {Efficient {{Iterative Linear-Quadratic Approximations}} for {{Nonlinear Multi-Player General-Sum Differential Games}}},
  booktitle = {2020 {{IEEE International Conference}} on {{Robotics}} and {{Automation}} ({{ICRA}})},
  author = {{Fridovich-Keil}, David and Ratner, Ellis and Peters, Lasse and Dragan, Anca D. and Tomlin, Claire J.},
  year = 2020,
  month = may,
  pages = {1475--1481},
  publisher = {IEEE},
  address = {Paris, France},
  doi = {10.1109/ICRA40945.2020.9197129},
  isbn = {978-1-7281-7395-5}
}

@article{benenati_linear-quadratic_2026,
  title = {Linear-{{Quadratic Dynamic Games}} as {{Receding-Horizon Variational Inequalities}}},
  author = {Benenati, Emilio and Grammatico, Sergio},
  year = {2026},
  journal = {IEEE Transactions on Automatic Control},
  volume={71},
  number={4},
  pages={2404-2417},
  doi={10.1109/TAC.2025.3632150}
}

@article{bemporad_explicit_2002,
  title = {The Explicit Linear Quadratic Regulator for Constrained Systems},
  author = {Bemporad, Alberto and Morari, Manfred and Dua, Vivek and Pistikopoulos, Efstratios N},
  year = 2002,
  journal = {Automatica},
  volume = {38},
  pages = {3--20}
}

@article{facchinei_generalized_2010,
  title = {Generalized {{Nash}} Equilibrium Problems},
  author = {Facchinei, Francisco and Kanzow, Christian},
  year = 2010,
  journal = {Annals of Operations Research},
  volume = {175},
  number = {1},
  pages = {177--211}
}

@article{fieni_game_2025,
	title={Game Theory in {{Formula}} 1: From Physical to Strategic Interactions},
	author={Fieni, Giona and Neumann, Marc-Philippe and Furia, Francesca and Caucino, Alessandro and Cerofolini, Alberto and Ravaglioli, Vittorio and Onder, Christopher H},
	journal={European Journal of Control},
	year={2026},
	publisher={Elsevier}
}

@inproceedings{zhu_sequential_2023,
  title = {A {{Sequential Quadratic Programming Approach}} to the {{Solution}} of {{Open-Loop Generalized Nash Equilibria}}},
  booktitle = {2023 {{IEEE International Conference}} on {{Robotics}} and {{Automation}} ({{ICRA}})},
  author = {Zhu, Edward L. and Borrelli, Francesco},
  year = 2023,
  month = may,
  pages = {3211--3217},
  publisher = {IEEE},
  address = {London, United Kingdom},
  doi = {10.1109/ICRA48891.2023.10160799},
  isbn = {979-8-3503-2365-8}
}

@article{le_cleach_lucidgames_2021,
  title = {{{LUCIDGames}}: {{Online Unscented Inverse Dynamic Games}} for {{Adaptive Trajectory Prediction}} and {{Planning}}},
  shorttitle = {{{LUCIDGames}}},
  author = {Le Cleac'h, Simon and Schwager, Mac and Manchester, Zachary},
  year = 2021,
  journal = {IEEE Robotics and Automation Letters},
  volume = {6},
  number = {3},
  pages = {5485--5492},
  issn = {2377-3766, 2377-3774},
  doi = {10.1109/LRA.2021.3074880}
}

@inproceedings{do_nascimento_game_2023,
  title = {A {{Game Theoretic Approach}} for {{Safe}} and {{Distributed Control}} of {{Unmanned Aerial Vehicles}}},
  booktitle = {2023 62nd {{IEEE Conference}} on {{Decision}} and {{Control}} ({{CDC}})},
  author = {Do Nascimento, Allan Andre and Papachristodoulou, Antonis and Margellos, Kostas},
  year = 2023,
  pages = {1070--1075},
  publisher = {IEEE},
  address = {Singapore, Singapore},
  doi = {10.1109/CDC49753.2023.10383672},
  isbn = {979-8-3503-0124-3}
}

@book{sastry_nonlinear_2013,
  title = {Nonlinear Systems: Analysis, Stability, and Control},
  author = {Sastry, S Shankar},
  year = 2013,
  publisher = {Springer Science \& Business Media}
}

@article{eckstein1998operator,
    title={Operator-splitting methods for monotone affine variational inequalities, with a parallel application to optimal control},
    author={Eckstein, Jonathan and Ferris, Michael C},
    journal={INFORMS Journal on Computing},
    volume={10},
    number={2},
    pages={218--235},
    year={1998},
    publisher={INFORMS}
}

@article{kim2018structure,
    title={A structure-preserving pivotal method for affine variational inequalities},
    author={Kim, Youngdae and Huber, Olivier and Ferris, Michael C},
    journal={Mathematical Programming},
    volume={168},
    number={1},
    pages={93--121},
    year={2018},
    publisher={Springer}
}

@article{baghbadorani2025douglas,
  author  = {Baghbadorani, Reza Rahimi and Benenati, Emilio and Grammatico, Sergio},
  title   = {A {{Douglas-Rachford}} splitting method for solving monotone variational inequalities in linear-quadratic dynamic games},
  journal = {IFAC-PapersOnLine},
  year    = {2026},
  note    = {23rd IFAC World Congress, Busan, Republic of Korea, August 23--28, 2026. To appear.},
}

@article{benenati2025explicit,
    title={The explicit game-theoretic linear quadratic regulator for constrained multi-agent systems},
    author={Benenati, Emilio and Belgioioso, Giuseppe},
    journal={arXiv preprint arXiv:2512.07749},
    year={2025}
}

@article{dirkse1995path,
    title={The {{PATH}} solver: a non-monotone stabilization scheme for mixed complementarity problems},
    author={Dirkse, Steven P and Ferris, Michael C},
    journal={Optimization {{Methods}} and {{Software}}},
    volume={5},
    number={2},
    pages={123--156},
    year={1995},
    publisher={Taylor \& Francis}
}

@article{arnstrom2022dual,
    title={A dual active-set solver for embedded quadratic programming using recursive {{${LDL}^T$}} updates},
    author={Arnstr{\"o}m, Daniel and Bemporad, Alberto and Axehill, Daniel},
    journal={IEEE Transactions on Automatic Control},
    volume={67},
    number={8},
    pages={4362--4369},
    year={2022},
    publisher={IEEE}
}

@inproceedings{arnstrom2024high,
    title={A high-performant multi-parametric quadratic programming solver},
    author={Arnstr{\"o}m, Daniel and Axehill, Daniel},
    booktitle={2024 IEEE 63rd Conference on Decision and Control (CDC)},
    pages={303--308},
    year={2024},
    organization={IEEE}
}

@article{bemporad2026nashopt,
  title = {{{NashOpt}}: A {{Python}} Library for Computing Generalized {{Nash}} Equilibria},
  shorttitle = {{{NashOpt}}},
  author = {Bemporad, Alberto},
  year = 2026,
  journal = {Optimization Methods and Software},
  volume = {41},
  number = {3},
  pages = {732--753},
  issn = {1055-6788, 1029-4937},
  doi = {10.1080/10556788.2026.2661736}
}

@book{patriksson_nonlinear_1999,
  title = {Nonlinear {{Programming}} and {{Variational Inequality Problems}}},
  author = {Patriksson, Michael},
  editor = {Pardalos, Panos M. and Hearn, Donald},
  year = 1999,
  series = {Applied {{Optimization}}},
  volume = {23},
  publisher = {Springer US},
  address = {Boston, MA},
  doi = {10.1007/978-1-4757-2991-7},
  isbn = {978-1-4419-4806-9 978-1-4757-2991-7}
}

@book{facchinei2003finite,
    title={Finite-dimensional variational inequalities and complementarity problems},
    author={Facchinei, Francisco and Pang, Jong-Shi},
    year={2003},
    publisher={Springer}
}

@inproceedings{mignoni_monviso_2025,
    title = {monviso: {A} {Python} {Package} for {Solving} {Monotone} {Variational} {Inequalities}},
    isbn = {978-3-907144-12-1},
    shorttitle = {monviso},
    doi = {10.23919/ECC65951.2025.11187100},
    urldate = {2025-12-08},
    booktitle = {2025 {European} {Control} {Conference} ({ECC})},
    publisher = {IEEE},
    author = {Mignoni, Nicola and Baghbadorani, Reza Rahimi and Carli, Raffaele and Esfahani, Peyman Mohajerin and Dotoli, Mariagrazia and Grammatico, Sergio},
    month = jun,
    year = {2025},
    pages = {1708--1713},
}

@article{liu2024input,
    title={Input-to-State Stability of {Newton} Methods in {{Nash}} Equilibrium Problems with Applications to Game-Theoretic Model Predictive Control},
    author={Liu, Mushuang and Kolmanovsky, Ilya},
    journal={arXiv preprint arXiv:2412.06186},
    year={2024}
}

@article{aghassi2006solving,
  title={Solving asymmetric variational inequalities via convex optimization},
  author={Aghassi, Michele and Bertsimas, Dimitris and Perakis, Georgia},
  journal={Operations Research Letters},
  volume={34},
  number={5},
  pages={481--490},
  year={2006},
  publisher={Elsevier}
}

@article{goulart2024clarabel,
  title={{Clarabel}: An interior-point solver for conic programs with quadratic objectives},
  author={Goulart, Paul J and Chen, Yuwen},
  journal={arXiv preprint arXiv:2405.12762},
  year={2024}
}

@incollection{murty2009computational,
  title={Computational complexity of complementary pivot methods},
  author={Murty, Katta G},
  booktitle={Complementarity and fixed point problems},
  pages={61--73},
  year={2009},
  publisher={Springer}
}
\end{document}